%% file: root_arXiv_accepted.tex
\documentclass[letterpaper, 10 pt, conference]{ieeeconf}  

\IEEEoverridecommandlockouts                              

\usepackage{cite}
\usepackage{amsmath}
\usepackage{amssymb}
\usepackage{amsfonts}
\usepackage{algorithmic}
\usepackage{mathdots}
\usepackage{graphicx}
\usepackage{textcomp}
\let\labelindent\relax \usepackage{enumitem}
\makeatletter
\let\IEEEproof\proof
\let\IEEEendproof\endproof
\let\proof\relax
\let\endproof\relax
\makeatother

\usepackage{amsthm} 

\makeatletter
\let\proof\IEEEproof
\let\endproof\IEEEendproof
\makeatother

\usepackage{relsize}

\usepackage{mathtools}
\usepackage[hidelinks]{hyperref} 
\usepackage{xurl}                
\input{macros}
\def\BibTeX{{\rm B\kern-.05em{\sc i\kern-.025em b}\kern-.08em
    T\kern-.1667em\lower.7ex\hbox{E}\kern-.125emX}}

\newcommand{\narrowbmatrix}[1]{%
  {\setlength{\arraycolsep}{2pt}\begin{bmatrix}#1\end{bmatrix}}}

\title{\LARGE \bf
Observer Design over Hypercomplex Quaternions
}
\author{Michael Sebek%
\thanks{*This work was co-funded by the European Union under the project ROBOPROX (reg. no. CZ$.02.01.01/00/22\_008/0004590$).}
\thanks{**Michael Sebek is with the Czech Technical University in Prague, 
Faculty of Electrical Engineering, 
Department of Control Engineering, 
Prague, 16000, Czech Republic
        {\tt\small michael.sebek@fel.cvut.cz}}%
}

\input{arxiv_accepted_stamp.tex}

\begin{document}

\ArxivAcceptedStampOverlay

\maketitle
\thispagestyle{empty}
\pagestyle{empty}

\begin{abstract}
We develop observer design over hypercomplex quaternions in a characteristic-polynomial-free framework. Using the standard right-module convention, we derive a right observable companion form and companion polynomial that encode error dynamics through right-eigenvalue similarity classes. We also give an Ackermann-type formula for real-coefficient target polynomials, where polynomial evaluation is similarity-equivariant. The resulting recipes place observer poles directly over quaternions and clarify when companion-coordinate updates and one-shot Ackermann formulas remain valid.
\end{abstract}

\section{INTRODUCTION}
\label{Sec:Introduction}
Quaternions extend \(\RR\) and \(\CC\) to the division algebra \(\HH\), but—unlike \(\RR\) and \(\CC\)—multiplication in \(\HH\) is \emph{noncommutative}. This single feature alters standard observer machinery: right and left eigenvalues must be distinguished, and (on the right) eigenvalues are defined only up to \emph{quaternionic similarity classes}. Familiar scalar proxies—determinants, characteristic polynomials, a full Cayley–Hamilton theorem—do not carry over verbatim; even polynomial evaluation depends on whether coefficients act on the left or on the right, so \emph{left and right zeros need not coincide}. As a result, several classical constructions require new formulations when one works natively over \(\HH\); see, e.g., quaternionic stability~\cite{PereiraVettori2006} and linear algebra over \(\HH\)~\cite{Rodman2014}.

In this paper, we develop an observer-focused, determinant–free framework built on \emph{observability} and its duality with controllability. We construct the observable companion (observability-form) realization directly from the observability matrix, link its right spectrum to companion-polynomial zeros, and derive a quaternionic \emph{Observer Ackermann formula} that assigns right–eigenvalue similarity classes using a desired real-coefficient polynomial. Throughout, we remain in the hypercomplex (scalar) model of \(\HH\) where signals, states, and operators live directly in \(\HH\). We work with the right spectrum and use left substitution for the companion polynomial—choices that make noncommutativity explicit and usable.

From the application viewpoint, quaternion-based models appear in several areas, but are often handled after conversion to real \(4\)-vector coordinates, especially in aerospace and robotics, where effective real-vector observers and filters are standard \cite{MarkleyCrassidis2014}. The purpose of the present paper is different: we show that observer design can also be developed \emph{directly over \(\HH\)} in a way that preserves the native hypercomplex structure. This matters precisely in settings where one wants to reason in terms of quaternionic spectra, similarity classes, companion forms, and central polynomial evaluation rather than through coordinate embeddings. 

Such a viewpoint is natural not only in quaternionic system theory itself, but also in areas where the hypercomplex structure is intrinsic or algorithmically meaningful, including quantum dynamics and control~\cite{daSilvaRouchon2008, MaLi25}, NMR/MRI pulse design~\cite{Emsley1992}, polarization/bivariate signal processing and color imaging~\cite{SPMag2023}, quaternion-valued filtering~\cite{TalebiWernerMandic20}, quantum computing~\cite{Meglicki2008}, and emerging quaternionic learning models~\cite{Takahashi2017, Parcollet2018}. Application-oriented quaternionic observer constructions for four-channel systems are discussed in~\cite{SebekEUROCAST2026}. 

We believe the observer results established here over \(\HH\)—including the observable companion construction and the quaternionic Observer Ackermann formula—are original.

\section{LINEAR ALGEBRA OVER QUATERNIONS}
\paragraph*{Quaternions}
A quaternion is a \emph{hypercomplex scalar}  \(q=a+b\i+c\j+d\k\) with \(a,b,c,d\in\RR\) and units \(\i,\j,\k\) satisfying
\(\i^{2}=\j^{2}=\k^{2}=\i\j\k=-1\).
The set \(\HH\) is a four-dimensional associative but \emph{noncommutative} division algebra over \(\RR\).
Write \(\Re(q)=a\), \(\Im(q)=b\i+c\j+d\k\), conjugate \(\bar{q}=a-b\i-c\j-d\k\), and norm \(|q|=\sqrt{q\bar q}\).

Other coordinate parameterizations are common. In the \emph{scalar–vector} form,
a quaternion is the real pair \(q_{\mathrm{sv}}=(v_{\mathrm{s}},v_{\mathrm{v}})\), where the scalar part is \(v_{\mathrm{s}}=a\) and
the vector part is \(v_{\mathrm{v}}=[\,b,\;c,\;d\,]^{\T}\in\RR^3\).
Equivalently, one stacks these into the vector–scalar \(4\)-vector
$q_{\mathrm{sv}}=\begin{bmatrix} v_{\mathrm{v}}^{\T} & v_{\mathrm{s}} \end{bmatrix}^{\T}\in\RR^4.$
These representations are standard in aerospace and robotics.

In this paper, we work natively
over \(\HH\) to preserve the noncommutative structure needed for right/left spectra and
observable–companion constructions.
For details, see~\cite{Rodman2014}.

\paragraph*{Similarity classes}
Quaternions \(q,r\) are \emph{similar}, denoted \(q\sim r\), if \(q=\alpha^{-1}r\alpha\) for some nonzero \(\alpha\in\HH\).
Similarity preserves \(\Re(\cdot)\), \(|\cdot|\), and \(|\Im(\cdot)|\) and partitions \(\HH\) into equivalence classes \([\cdot]\).
Real quaternions are isolated: \([a]=\{a\}\).
If \(\Im(q)\neq0\), then \([q]=\{\Re(q)+u|\Im(q)|: u\in\mathrm{span}\{\i,\j,\k\},|u|=1\}\), a 2-sphere meeting the standard complex line \(\{a+b\i: a,b\in\RR\}\) in the two points \(\Re(q)\pm \i|\Im(q)|\).
For instance, \([\k]\) contains all imaginary units, including \(\pm \i\).

\paragraph*{Matrices over \(\HH\)}
For \(m,n\in\NN\), let \(\HHmn\) denote the set of \(m\times n\) matrices with entries in $\HH.$
We write $A^{*}=\bar{A}^{\T}$ (entry-wise conjugation followed by transpose).
Noncommutativity affects basic identities: in general 
\( \overline{AB}\neq \bar{A}\bar{B}\) and \((AB)^{\T}\neq B^{\T}A^{\T}\), while the mixed rule \((AB)^{*}=B^{*}A^{*}\) does hold.
See~\cite{Rodman2014} for further details.

\emph{Caveat for design.} Classical $\RR/\CC$ tools—determinant, characteristic polynomial, and the Cayley–Hamilton theorem—do not carry over verbatim to $\HH$. The determinant admits nonstandard variants with different meanings; a “characteristic polynomial” in the usual sense may be absent or enjoy only limited properties; and Cayley–Hamilton can fail or requires alternative formulations. Hence these tools cannot be used naively for observer design over $\HH$ and are replaced here by determinant-free, characteristic-polynomial-free constructions.

\paragraph*{Rank and invertibility}
For \(A\in\HHmn\), the quaternionic rank \(\operatorname{rank}_{\HH}A\) is the maximal number of \emph{right–independent} columns (i.e., if \(\sum_j a_j c_j=0\) with \(c_j\in\HH\) on the right, then all \(c_j=0\)).
A square \(A\in\HH^{n\times n}\) is invertible iff \(\operatorname{rank}_{\HH}A=n\)~\cite{Rodman2014}.

\paragraph*{Right eigenvalues and the spectrum}
For \(A\in\HH^{n\times n}\), a \emph{right eigenpair} \((v,\lambda)\) satisfies
\begin{IEEEeqnarray}{rCl}
A v &=& v\lambda,\qquad v\in\HH^{n}\setminus\{0\},\ \lambda\in\HH.
\end{IEEEeqnarray}
Right scaling of \(v\) by \(\alpha\neq0\) replaces \(\lambda\) by \(\alpha^{-1}\lambda\alpha\), so right eigenvalues are defined only up to similarity.
The right spectrum is thus the set of classes
$\sigma_{\err}(A):=\{[\lambda]:\exists v\neq0 \text{ with } Av=v\lambda\}\subset \HH/\!\sim,$
which is invariant under similarity \(A\mapsto S^{-1}AS\)~\cite{Rodman2014}.
When needed, take the unique complex representative with nonnegative imaginary part—the \emph{standard eigenvalue}.

\paragraph*{Polynomials and evaluation}
Let \(\HH[\lambda]\) be the ring of quaternionic polynomials \(p(\lambda)=\sum_{k=0}^{n}p_k\lambda^{k}\) with \(p_k\in\HH\) and indeterminate \(\lambda\) commuting with all coefficients.
Evaluation at \(q\in\HH\) depends on side:
the \emph{right value} is \( p(q)_{\err}=\sum_{k=0}^{n}p_k q^{k} \),
while the \emph{left value} is \( p(q)_{\ell}=\sum_{k=0}^{n} q^{k} p_k \).

\paragraph*{Left/right zeros}
A \emph{right zero} of \(p\) is \(q\) with \(p(q)_{\err}=0\); a \emph{left zero} satisfies \(p(q)_{\ell}=0\).
Left and right zeros correspond up to similarity:
if \(p(q)_{\ell}=0\), then for some \(\alpha\neq0\), \(p(\alpha^{-1}q\alpha)_{\err}=0\), and conversely~\cite{GordonMotzkin65}.
If all \(p_k\in\RR\) (central), then \(p(q)_{\err}=p(q)_{\ell}\) and the two zero sets coincide.

\paragraph*{Isolated vs.\ spherical zeros}
A degree-\(n\) quaternionic polynomial has either \(n\) (counting multiplicity) right zeros or infinitely many.
A nonreal zero is either \emph{isolated} (a single point in its class) or \emph{spherical} (the full class \([q]\), a 2-sphere).
If the coefficients are real, every zero is either real (hence isolated) or spherical~\cite{PogoruiShapiro2004}.

\paragraph*{Notation}
We regard \(\HH^{n}\) as a right \(\HH\)-module.
Let \(e_i\) be the \(i\)th canonical basis vector of \(\RR^{n}\subset\HH^{n}\) (a \(1\) in position \(i\), zeros elsewhere); in particular, \(e_n=[0,\ldots,0,1]^{\T}\).
Because these entries are real (central), left/right scalar multiplication is unambiguous whenever \(e_i\) appears.

\section{QUATERNIONIC STATE–SPACE SISO SYSTEMS}
We study continuous-time linear time-invariant (LTI) single-input single-output (SISO) systems over the division algebra~\(\HH\).

\paragraph*{Model}
In state-space form, the dynamics are
\begin{equation}\label{eq:ct-ss-siso}
\begin{aligned}
\dot x(t) &= Ax(t) + Bu(t),\\
y(t)     &= Cx(t) + Du(t),
\end{aligned}
\end{equation}
where \(A\in\HH^{n\times n}\), \(B\in\HH^{n}\) (column), \(C\in\HH^{1\times n}\) (row), \(D\in\HH\),
the state vector is \(x(t)\in\HH^{n}\), and the input/output scalars are \(u(t),y(t)\in\HH\).

\paragraph*{Conventions}
We adopt the \emph{right-module} convention: matrices act from the left and quaternionic scalars multiply signals on the \emph{right}. Thus the products \(Bu(t)\) and \(Du(t)\) are right-multiplications by \(u(t)\); this guarantees right-linearity in both \(x\) and \(u\).

\paragraph*{Internal stability}
With $u\!\equiv\!0$, stability means $x(t)\!\to\!0$ (CT) or $x(k)\!\to\!0$ (DT) for all initial states.
Equivalently, using the standard right eigenvalues $\lambda$ of $A$, the system is stable iff $\Re\lambda<0$ for all $\lambda$; see~\cite{PereiraVettori2006}.

\paragraph*{Controllability}
As in the real/complex case, controllability means the ability to steer the state from any $x_0\in\HH^n$ to any $x_f\in\HH^n$ within a finite horizon by a suitable input.
In the single–input setting, the \emph{controllability matrix} is
\begin{IEEEeqnarray*}{rCl}
\mathcal C &=& \bigl[ B  AB  \cdots  A^{n-1}B \bigr] \in \HH^{n\times n},
\end{IEEEeqnarray*}
and the pair $(A,B)$ is controllable iff $\mathcal C$ is invertible over $\HH$
(equivalently, iff $\operatorname{rank}_{\HH}\mathcal{C} = n$) \cite{Jiang2020,Rodman2014,Zhang1997}.

\paragraph*{Observability}
Dually, observability means that the initial state can be uniquely reconstructed from finitely many output measurements.
In the single–output setting, the \emph{observability matrix} is
\begin{IEEEeqnarray}{rCl}
\mathcal{O} &=& 
\begin{bmatrix}
C\\
CA\\
\vdots\\
CA^{n-1}
\end{bmatrix}
\in \HH^{n\times n},
\label{Eq0OBSmat}
\end{IEEEeqnarray}
and the pair $(A,C)$ is observable iff $\mathcal{O}$ is invertible over $\HH$
(equivalently, iff $\operatorname{rank}_{\HH}\mathcal{O} = n$) \cite{Jiang2020,Rodman2014,Zhang1997}.

\paragraph*{Similarity}
All rank conditions are over $\HH$ and are invariant under similarity:
$A\mapsto S^{-1}AS$, $B\mapsto S^{-1}B$, $C\mapsto CS$ for any invertible $S\!\in\!\HH^{n\times n}$.

\section{OBSERVABLE COMPANION FORM AND COMPANION POLYNOMIAL}
In the classical (real/complex) LTI setting, two staple tools are the observable companion form and the characteristic or minimal polynomial. For quaternionic models, the observable companion form remains valid. By contrast, characteristic and minimal polynomials do not translate verbatim, since quaternions lack a straightforward determinant. Consequently, we replace them with a suitably defined quaternionic companion polynomial.

\begin{Def}[\emph{Observable companion form and companion polynomial}]
Let $A_o\in\HHnn$ and $C_o\in\HH^{1\times n}$ be
\begin{IEEEeqnarray}{rCl}
A_o &=& 
\begin{bmatrix}
 0 &\ldots& 0 & -a_{0}\\
1 & \ldots & 0 & -a_{1}\\
\vdots&\ddots&\vdots&\vdots\\
0 & \ldots & 1 & -a_{n-1}
\end{bmatrix},\\
C_o &=& \begin{bmatrix} 0 & 0 & \ldots & 1 \end{bmatrix},
\label{Eq:AoCoObsComp}
\end{IEEEeqnarray}
with $a_0,a_1,\ldots,a_{n-1}\in\HH$. We say that $(A_o,C_o)$ is in right \emph{observable companion (canonical) form}. The \emph{companion polynomial} associated with $(A_o,C_o)$ (or with $A_o$) is the monic quaternionic polynomial
\begin{IEEEeqnarray*}{rCl}
a(\lambda) &=& a_0 + a_1 \lambda + \cdots  + a_{n-1} \lambda^{n-1} + \lambda^n \in \HH[\lambda].
\end{IEEEeqnarray*}
\end{Def}

If \(A\in\HH^{n\times n}\) is \emph{cyclic} with respect to the row $C$—i.e., there exists \( C\in\HH^{1\times n}\) with 
the observability matrix \(\mathcal{O}\) invertible—then \(A\) is similar to the right observable companion form. The next theorem gives a constructive procedure.
\begin{The}[\emph{Determinant-free companion form}]
\label{Th:DeterminantFreeCompanion}
Let \(A\in\HH^{n\times n}\) and \(C\in\HH^{1\times n}\). Assume the SISO observability matrix
\(\mathcal{O}\) from \eqref{Eq0OBSmat} is invertible over \(\HH\).
\begin{enumerate}[leftmargin=*,noitemsep]
\item There exists an invertible similarity \(T\in\HH^{n\times n}\) such that
\[
A_o \;=\; T^{-1}AT,\qquad C_o \;=\; CT,
\]
and the pair \((A_o,C_o)\) is in the right observable companion form
\begin{equation}\label{Eq0RightObsForm-NEW}
\begin{IEEEeqnarraybox}[][c]{rCl}
A_o &=&
\begin{bmatrix}
0 & 0 & \cdots & -a_{0}\\
1 & 0 & \cdots & -a_{1}\\
\vdots & \ddots & \ddots & \vdots\\
0 & \cdots & 1 & -a_{n-1}
\end{bmatrix},\\
C_o &=& 
\begin{bmatrix}
0&0&\cdots&1
\end{bmatrix},
\end{IEEEeqnarraybox}
\end{equation}
where the scalars \(a_0,\ldots,a_{n-1}\in\HH\) are \emph{defined} as the (negatives of the) entries in the
rightmost column of \(A_o\).

\item The corresponding similarity matrix is uniquely determined as
\begin{IEEEeqnarray}{rCl}
s &:=& \mathcal{O}^{-1} e_{n}\ \in\ \HH^{n}, \IEEEnonumber\\
T &:=&
\begin{bmatrix}
s & As & \cdots & A^{n-1}s
\end{bmatrix},
\label{Eq:T-columns-rightmost-NEW}
\end{IEEEeqnarray}

\item For the given pair \((A,C)\), the matrices \(A_o,C_o\) constructed via \eqref{Eq:T-columns-rightmost-NEW} and the coefficients \(a_k\) (as defined above) are unique.
\end{enumerate}
\end{The}

\begin{proof}
We adapt the classical real/complex argument (cf.~\cite{AntsaklisMichel06}) to the quaternionic setting, avoiding determinants and Cayley–Hamilton.

\emph{Choice of basis and similarity.}
Let \(s:=\mathcal{O}^{-1}e_n\) and \(T=[\,s\ As\ \cdots\ A^{n-1}s\,]\).
From \(\mathcal{O} s=e_n\) we obtain \(Cs=\cdots=CA^{n-2}s=0\) and \(CA^{n-1}s=1\), hence
\[
CT=\bigl[\,Cs\ \ CAs\ \ \cdots\ \ CA^{n-1}s\,\bigr]=e_n^{\T}.
\]
Moreover, for each \(k=0,\ldots,n-1\), the row \(CA^{k}T\) has the
triangular pattern
\begin{equation}\label{Eq:triangular-pattern}
CA^{k}T=\bigl[\,\underbrace{0\ \cdots\ 0}_{n-k-1}\ \ 1\ \ \star \ \cdots\ \star\,\bigr],
\end{equation}
where `$\star$' denotes unspecified entries.
Hence, \((CA^{k}T)e_j=0\) for \(j<n-k\) and \((CA^{k}T)e_{n-k}=1\).

\emph{Define the transformed pair.}
Set \(M:=T^{-1}AT\) and \(C_o:=CT=e_n^{\T}\).

\emph{Output row.}
From \(CT=e_n^{\T}\) we have \(C_o=e_n^{\T}=[\,0\ \cdots\ 0\ 1\,]\).

\emph{Subdiagonal shift property.}
Write \(T=[w_1,\ldots,w_n]\) with \(w_1=s\) and \(w_{j+1}=Aw_j\) by construction.
Then
\begin{IEEEeqnarray*}{c}
    M e_j = T^{-1}AT e_j = T^{-1}Aw_j = T^{-1}w_{j+1} = e_{j+1},
\end{IEEEeqnarray*}
for $j=1,\ldots,n-1,$
so \(M\) has ones on the subdiagonal and zeros below it.

\emph{Rightmost column and the definition of \(a_k\).}
By definition of the \emph{observable companion form}, the last column of \(M\) is some vector
\(-[a_0,\ldots,a_{n-1}]^{\T}\in\HH^{n}\). We \emph{define} the coefficients \(a_k\in\HH\) to be these
entries. With the output row and subdiagonal shift established above, this uniquely identifies \(M\) as the observable companion \(A_o\) and yields \eqref{Eq0RightObsForm-NEW}.

\emph{Uniqueness.}
The vector \(s\) is uniquely determined by \(\mathcal{O} s=e_n\), hence \(s=\mathcal{O}^{-1}e_n\).
Then the recursion \(w_{j+1}=Aw_j\) uniquely fixes the columns of \(T=[\,s,As,\ldots,A^{n-1}s\,]\).
With \(C_o=e_n^{\T}\), subdiagonal ones and the last column \(-[a_0,\ldots,a_{n-1}]^{\T}\), the matrix
\(A_o=M\) is unique as well.
\end{proof}
\begin{Rem}
In the real/complex (central) case, the last-column coefficients $a_k$ also satisfy 
$CA^{n}=-\sum_{k=0}^{n-1}(CA^{k})\,a_k$. Over $\HH$, a right–linear expansion still holds but with (generally different) scalars $b_k\in\HH$:
$CA^{n}=-\sum_{k=0}^{n-1}(CA^{k})\,b_k$, and there is no simple closed-form relation between $a_k$ and $b_k$ due to noncommutativity; if the $a_k$ are real, then $b_k=a_k$.
\end{Rem}

\begin{Rem}
Compared with the classical \(\RR/\CC\) case, the construction proceeds \emph{in reverse}.
Instead of first computing a characteristic or minimal polynomial (via determinants) and then forming a companion realization, 
we directly build the \emph{observable companion form} from the invertible \emph{observability matrix}~\(\mathcal{O}\), 
and only afterwards \emph{read off} the coefficients of the companion polynomial from the rightmost column of \(A_o\).
This approach is determinant–free and avoids the use of characteristic polynomials, Cayley–Hamilton, or adjugates,
whose straightforward generalizations to~\(\HH\) are problematic.
\end{Rem}
\begin{Ex}\label{ExCompanion}
Let
\begin{IEEEeqnarray}{rCl} \label{Eq0ACEx}
A&=&\tfrac14 \begin{bmatrix} -2+\j  &   \i  \\
     \i   &   -2-\j \end{bmatrix},\;\;
C=\begin{bmatrix} \j & \k \end{bmatrix}.
\end{IEEEeqnarray}
Then the observability matrix inverse is
\begin{IEEEeqnarray}{c}
    \mathcal{O}^{-1} = \tfrac12
    \begin{bmatrix}
    -1-3\j  &  -4\j \\
     \i+\k   &   4\k
    \end{bmatrix}. \label{Eq0OBSi} \IEEEyesnumber
\end{IEEEeqnarray}
and the similarity transformation matrix \eqref{Eq:T-columns-rightmost-NEW} reads
\begin{IEEEeqnarray}{c}
    T = \tfrac12
    \begin{bmatrix}
    -4\j  &   1+\j \\ 
     4\k  &  -\i-3\k
    \end{bmatrix}. \label{Eq0T} \IEEEyesnumber
\end{IEEEeqnarray}
The right observable companion form of $A$ is
{\setlength{\arraycolsep}{2pt}
\begin{IEEEeqnarray}{rCl}\label{Eq0ExAoCo}
A_o &=& 
\begin{bmatrix}
   0  &  -0.25+0.25\j \\
   1  &  -1+0.5\j   
\end{bmatrix}, \;
C_o=\begin{bmatrix} 0 & 1 \end{bmatrix},  
\end{IEEEeqnarray}
}
and the associated companion polynomial is
\begin{IEEEeqnarray}{rCl}\label{Eq0Excomppol}
a(\lambda) &=& 0.25-0.25\j + (1-0.5\j)\lambda + \lambda^2.
\end{IEEEeqnarray}
\end{Ex}
We obtain an equally clean annihilating identity for the right \emph{observable} companion, again without invoking a full Cayley–Hamilton theorem over $\HH$ (cf.~\cite{ChapmanMachen16, Rodman2014}). We will use left substitution for the companion polynomial at $A_o.$
\begin{The}[\emph{Companion polynomial annihilates its observable companion matrix over $\HH$}]
\label{Th:ObsCompanionAnnihilates}
Let $A_o\in\HH^{n\times n}$ be in observable companion form
\begin{IEEEeqnarray}{rCl}
A_o &=&
\begin{bmatrix}
0 & 0 & \cdots & -a_{0}\\
1 & 0 & \cdots & -a_{1}\\
\vdots & \ddots & \ddots & \vdots\\
0 & \cdots & 1 & -a_{n-1}
\end{bmatrix},
\label{eq:Ao-obs-form}
\end{IEEEeqnarray}
with $a_0,\ldots,a_{n-1}\in\HH$, and define the companion polynomial
$a(\lambda) = a_0 + a_1\lambda + \cdots + a_{n-1}\lambda^{n-1} + \lambda^{n} \ \in \HH[\lambda].$
Then $a(\lambda)$ annihilates $A_o$ under \emph{left substitution}:
\begin{IEEEeqnarray}{c}
Ia_0  +  A_o a_1 + \cdots + A_o^{n-1} a_{n-1}+ A_o^{n} = 0.
\label{eq:left-annihilation-obs}
\end{IEEEeqnarray}
\end{The}
\begin{proof}
\emph{Setup and evaluation convention.}
Let $A_o$ be in the observable companion form \eqref{eq:Ao-obs-form} with last column
$-\bigl[a_0,\ldots,a_{n-1}\bigr]^{\T}$, where $a_k\in\HH$.
Following Section~2, we evaluate polynomials at $A_o$ by \emph{left substitution}:
\begin{IEEEeqnarray}{c}
a(A_o)_\ell=A_o^{n}+\sum_{k=0}^{n-1}A_o^{k}\,a_k.
\end{IEEEeqnarray}
Our goal is to show $a(A_o)_\ell=0$.

\emph{Krylov columns and the shift property.}
Define $u_k:=A_o^{k}e_1$ for $k\ge 0$. Since $A_o$ has ones on the subdiagonal,
\[
A_o e_j=e_{j+1}\quad (j=1,\ldots,n-1),
\]
hence $u_k=e_{k+1}\,\;(k=0,\ldots,n-1)$, and thus
\begin{IEEEeqnarray}{c}
  \mathcal{K}:=
  \begin{bmatrix}
      u_0&\ldots&u_{n-1}
  \end{bmatrix}=I_n . 
\end{IEEEeqnarray}

\emph{The step-$n$ recurrence from the last column.}
Reading off the last column of $A_o$ gives
\[
A_o^{n}e_1
=
-\begin{bmatrix}a_0\\ a_1\\ \vdots\\ a_{n-1}\end{bmatrix}
=
-\sum_{k=0}^{n-1} u_k\,a_k .
\]
Left-multiplying by $A_o^{j}$ and using $A_o^{j}u_k=u_{k+j}$ yields, for all $j\ge 0$,
\begin{equation}\label{eq:recurrence-left-eval}
u_{n+j}
=
-\sum_{k=0}^{n-1} u_{k+j}\,a_k .
\end{equation}

\emph{Annihilation under left substitution.}
For any $j\ge 0$,
\begin{IEEEeqnarray}{rCl}
a(A_o)_\ell\,u_j
&=&
\Bigl(A_o^{n}+\sum_{k=0}^{n-1}A_o^{k}a_k\Bigr)u_j \IEEEnonumber \\
&=&
u_{n+j}+\sum_{k=0}^{n-1}u_{k+j}\,a_k
=
0
\end{IEEEeqnarray}
by \eqref{eq:recurrence-left-eval}. In particular, this holds for $j=0,1,\ldots,n-1$, so
\[
a(A_o)_\ell\;\mathcal K = 0.
\]
Since $\mathcal K=I_n$, we conclude $a(A_o)_\ell=0$, which is the claimed annihilation under left substitution.
\end{proof}
\begin{Rem}
This is \emph{not} a quaternionic Cayley–Hamilton theorem (which generally fails over~\(\HH\)).
Rather, it is a \emph{one-sided} annihilation under left substitution that hinges on the observable companion’s
shift/subdiagonal structure and the rightmost-column read-off of the coefficients.
It neither asserts a two-sided polynomial identity nor identifies \(a(\lambda)\) with any
characteristic/minimal polynomial; it also does \emph{not} automatically extend to arbitrary
similar realizations outside the observable companion form.
\end{Rem}

The right spectrum of the \emph{observable} companion matrix is likewise composed of the similarity classes of the right zeros of its companion polynomial.

\begin{The}[\emph{Right zeros form the right spectrum of the observable companion matrix}]
\label{Th:Ao-roots-spectrum}
Let $A_o\in\HH^{n\times n}$ be in observable companion form~\eqref{eq:Ao-obs-form}, and let
$a(\lambda) = a_0 + a_1\lambda + \cdots + a_{n-1}\lambda^{n-1} + \lambda^{n}
\in \HH[\lambda]$
be its companion polynomial (with $a_k\in\HH$ defined by the last column of $A_o$). Then:
\begin{enumerate}[leftmargin=*,noitemsep]
\item If $\lambda\in\HH$ is a \emph{right} zero of $a(\lambda)$,
then $\lambda$ is a \emph{right } eigenvalue of $A_o$. The corresponding eigenvector
$v = \begin{bmatrix} v_1,\ldots,v_n\end{bmatrix}$
is obtained by the backward recurrence
\begin{IEEEeqnarray}{c}\label{Eq0BcwrdRec}
v_n=\!1;\;v_{j-1}\!=\!a_{j-1}\!+\!v_j\lambda ,\;j=n,n-1,\ldots,2,   
\end{IEEEeqnarray}
and satisfies $A_o v = v\lambda$.
\item Hence, the right spectrum of $A_o$ consists exactly of the similarity classes of the right zeros of $a(\lambda)$.
\item Similarly, every \emph{left zero} $\mu$ of $a(\lambda)$ is (up to similarity) also a right eigenvalue of $A_o$.
\end{enumerate}
\end{The}
\begin{proof}
\emph{(1) Right zeros $\Rightarrow$ right eigenpairs.}
Assume $a(\lambda)_\err=0$ and construct $v\in\HH^n$ by the backward recurrence \eqref{Eq0BcwrdRec}. We claim $A_o v=v\lambda$.

For rows $j=2,\ldots,n$ of $A_o v$, the observable companion structure gives
\begin{IEEEeqnarray*}{rCl}
(A_o v)_j &=& v_{j-1} - a_{j-1}\,v_n\\
&=& (v_j\lambda + a_{j-1}) - a_{j-1}
= v_j\lambda
= (v\lambda)_j.
\end{IEEEeqnarray*}
For the first row,
\begin{IEEEeqnarray*}{rCl}
(A_o v)_1 &=& -a_0\,v_n \;=\; -a_0,\\
(v\lambda)_1 &=& v_1\lambda
= \bigl(a_1+\cdots+a_{n-1}\lambda^{n-2}+\lambda^{n-1}\bigr)\lambda\\
&=& a_1\lambda+\cdots+a_{n-1}\lambda^{n-1}+\lambda^{n}.
\end{IEEEeqnarray*}
Hence $(A_o v)_1=(v\lambda)_1$ is equivalent to
$a_0 + \sum_{k=1}^{n-1} a_k \lambda^k + \lambda^n=0$, i.e.\ $a(\lambda)_\err=0$,
which holds by assumption. Therefore $A_o v=v\lambda$ and $\lambda$ is a right eigenvalue.

\emph{(2) Right eigenpairs $\Rightarrow$ zeros up to similarity.}
Conversely, suppose $A_o x = x\mu$ with $x\neq 0$.
From rows $j=2,\ldots,n$ of the observable companion,
\[
x_{j-1} - a_{j-1}x_n = x_j\mu,
\]
which rearranges to $x_{j-1}= a_{j-1}x_n+x_j\mu$. Iterating backwards yields, for every $k=1,\ldots,n-1,$
\begin{IEEEeqnarray*}{c}
x_k = a_k x_{n} + a_{k+1}x_{n} \mu+\cdots+a_{n-k}x_{n} \mu^{n-k-1}+x_{n}\mu^{n-k}.
\end{IEEEeqnarray*}
In particular, for $k=1$
\begin{IEEEeqnarray*}{c}
x_1 = a_1x_n + a_2x_n\mu+\cdots+a_{n-1}x_n\mu^{n-2}+x_n\mu^{n-1}.
\end{IEEEeqnarray*}
Using the first row of $A_o x=x\mu$ and substituting for $x_1$,
\begin{IEEEeqnarray*}{c}
-a_0x_n
= x_1\mu
= a_1x_n\mu + \cdots+a_{n-1}x_n\mu^{n-1}+x_n\mu^{n}
\end{IEEEeqnarray*}
Note that $x_n\neq 0$; otherwise, from the $n$-th row $x_{n-1}-a_{n-1}x_n=x_n\mu$ we get $x_{n-1}=0$, and iterating upward forces $x=0$, a contradiction. Hence set the similar quaternion
\[
\lambda:=x_n\mu x_n^{-1},
\]
so that $x_n\mu^k=\lambda^k x_n$ for all $k$. Then
\begin{IEEEeqnarray*}{c}
-a_0x_n
= a_1\lambda x_n+\cdots+  a_{n-1}\lambda^{n-1} x_n+\lambda^{n} x_n,   
\end{IEEEeqnarray*}
and right–cancelling $x_n$ gives
\[
a_0 + a_1\lambda+\cdots+  a_{n-1}\lambda^{n-1} +\lambda^{n} = 0,
\]
i.e.\ $a(\lambda)_\err=0$.
Thus every right eigenvalue $\mu$ is \emph{similar} to a right zero $\lambda$, proving the spectral statement.

\emph{(3) Left zeros.}
By Gordon–Motzkin, left and right zeros of a quaternionic polynomial lie in the same similarity classes;
hence every left zero is similar to a right zero and therefore (up to similarity) to a right eigenvalue of $A_o$.

Combining (1)–(3) gives the three claims of the theorem.
\end{proof}

\begin{Ex}
For $A_o$ from \eqref{Eq0ExAoCo} with $a(\lambda)$ from \eqref{Eq0Excomppol}, we have
\begin{IEEEeqnarray*}{rCl}
a(A_o)_\ell &=&  
I_2(0.25\!-\!0.25\j) +
\begin{bmatrix}
   0  &  \!-\!0.25\!+\!0.25\j \\
   1  &  \!-\!1\!+\!0.5\j   
\end{bmatrix}
(1\!-\!0.5\j)\\
&&\negmedspace{}+
\begin{bmatrix}
\!-\!0.25\!+\!0.25\j & 0.125\!-\!0.375\j \\
    \!-\!1\!+\!0.5\j & 0.5\!-\!0.75\j     
\end{bmatrix}
=
\begin{bmatrix}
    0&0\\0&0
\end{bmatrix}.
\end{IEEEeqnarray*}
\end{Ex}

The companion polynomial can likewise be attached to any \emph{cyclic} quaternionic matrix $A$
via its \emph{observable} companion realization $(A_o,C_o)$.

\begin{Def}[\emph{Companion polynomial of an observable pair}]
Let $(A,C)$ with $A\in\HH^{n\times n}$ and $C\in\HH^{1\times n}$ be such that the
observability matrix $\mathcal{O}$ from \eqref{Eq0OBSmat}
is invertible over $\HH$.
Let $(A_o, C_o)$ be the unique right observable companion realization of $(A, C)$ (Theorem~\ref{Th:DeterminantFreeCompanion}),
and write the last column of $A_o$ as $-\bigl[a_0,\ldots,a_{n-1}\bigr]^{\T}$ with $a_k\in\HH$.
The \emph{companion polynomial} of $(A,C)$ is the monic quaternionic polynomial
\begin{IEEEeqnarray}{rCl}
a(\lambda) &=& a_{0}+a_{1}\lambda+\cdots+a_{n-1}\lambda^{n-1}+\lambda^{n}.
\label{Eq:obs-comppoly}
\end{IEEEeqnarray}
\end{Def}
\begin{Rem}
Throughout, we use the right observable companion form. Other orientations—left/upper/lower variants obtained by permutations—are equivalent and lead,
after the obvious similarity/permutation, to the \emph{same} monic companion polynomial.
\end{Rem}
\begin{The}[\emph{Right zeros form the right spectrum of $A$ (observable case)}]
\label{Th:AC-roots}
Let $(A,C)$ be an observable quaternionic SISO pair and let $a(\lambda)$
be its companion polynomial from \eqref{Eq:obs-comppoly}. Then
\begin{enumerate}[leftmargin=*,noitemsep]
\item Any right root $\lambda\in\HH$ of $a(\lambda)$ yields a right eigenvalue (up to similarity) of $A$.
\item Hence, the right spectrum of $A$ consists exactly of the similarity classes of the right zeros of $a(\lambda)$.
\item Likewise, each left root $\mu$ of $a(\lambda)$ is similar to some right eigenvalue of $A$.
\end{enumerate}
\end{The}
\begin{proof}
By Theorem~\ref{Th:DeterminantFreeCompanion}, $A$ is similar over $\HH$ to its observable companion $A_o$, so $\sigma_\err(A)=\sigma_\err(A_o)$. The claims then follow from Theorem~\ref{Th:Ao-roots-spectrum}.
\end{proof}

\begin{Rem}
As in the controllable case, the companion polynomial shares several roles of the
characteristic/minimal polynomial but not all: in general, it \emph{does not} annihilate $A$;
it only annihilates the observable companion $A_o$. Hence $a(A_o)_\ell=0$ while typically $a(A)_\ell\neq 0$.
\end{Rem}

\begin{Rem}
When \(A\) is cyclic, it is natural to speak of the \emph{companion polynomial of \(A\)}
constructed from any observable SISO pair \((A,C)\).
Different cyclic choices of \(C\) may yield different polynomials. Still, their right zeros agree up to quaternionic similarity classes: for every such \(C\), the right zeros of the associated polynomial determine the same set of similarity classes, namely the right spectrum of \(A\).
For a fixed pair \((A,C)\), however, the observable companion \((A_o,C_o)\) and the attached polynomial are unique.
\end{Rem}

\section{OBSERVERS}
An observer is a dynamic estimator that reconstructs the state $\hat x$ from known inputs \(u\) and measured outputs \(y\) by injecting the innovation \(y-C\hat x\) through a gain \(L\). It enables implementable state feedback \(u=K\hat x\), monitoring, and prediction. For quaternionic systems, the same observer structure is used, but pole assignment must be interpreted in terms of right-eigenvalue similarity classes.

A well-known \emph{Luenberger full-order observer}  \cite{Luenberger1966} is described by
\begin{IEEEeqnarray}{rCl}
\dot{\hat x} &=& A\hat x + Bu + L(y - C\hat x), \label{eq:ct-fo}\\
e &:=& x-\hat x,\quad \dot e = (A-LC)e. \IEEEnonumber\label{eq:ct-fo-err}
\end{IEEEeqnarray}
and the design goal is to place eigenvalues of the observer error dynamics (observer poles) $\sigma_\err (A-LC)$ in the open left half-plane. 
%

\section{DUALITY OVER \texorpdfstring{$\HH$}{HH}}
Over $\RR$, controller--observer duality is expressed through transpose: controllability of $(A,B)$ is equivalent to observability of $(A^{\T},B^{\T})$, and one typically transfers gains via $L=K^{\T}$; see, e.g.,~\cite{AntsaklisMichel06}.

Over $\HH$, the transpose is no longer the correct dual object. The appropriate operation is the \emph{conjugate transpose} $M^*=\overline{M}^{\T}$, since it preserves the quaternionic Hermitian form,
\(
\langle Ax,y\rangle=\langle x,A^*y\rangle,
\)
and is compatible with quaternionic spectral theory~\cite{Zhang1997,Rodman2014,FarenickPidkowich2003,AlpayColomboKimsey2016}. Thus the duality mapping must be written as
\begin{IEEEeqnarray*}{c}
A \leftrightarrow A^*, \qquad
B \leftrightarrow C^*, \qquad
C \leftrightarrow B^*, \qquad
D \leftrightarrow D^* .
\end{IEEEeqnarray*}
Accordingly, the original system and its dual are linked by
\begin{IEEEeqnarray*}{rCl}
\mathcal{O}_n(A^*,B^*) &=& \mathcal{C}_n(A,B)^* ,
\end{IEEEeqnarray*}
because \((AB)^*=B^*A^*\); see~\cite{Zhang1997,Rodman2014}. Since rank is preserved under $^*$, this yields
\begin{IEEEeqnarray*}{c}
(A,B)\ \text{controllable}\ \Longleftrightarrow\ (A^*,B^*)\ \text{observable}, \\
(A,C)\ \text{observable}\ \Longleftrightarrow\ (A^*,C^*)\ \text{controllable}.
\end{IEEEeqnarray*}
Therefore, if $K\in\HH^{1\times n}$ is a state-feedback gain for the controllable dual pair $(A^*,C^*)$, then the corresponding observer gain for the original system is
\begin{IEEEeqnarray*}{rCl}
L &=& K^* \in \HH^{n\times 1}.
\end{IEEEeqnarray*}
Moreover,
\[
\sigma_\err(A-LC)=\sigma_\err(A^*-C^*K),
\]
because the adjoint converts right eigenvalues into conjugates, and quaternionic similarity classes satisfy $[q]=[\overline q]$; see~\cite{FarenickPidkowich2003,AlpayColomboKimsey2016}. Hence both matrices have the same assigned right-eigenvalue similarity classes. See also quaternion-valued system treatments in~\cite{Jiang2020,ChenFecKanWang2022}.

\section{QUATERNIONIC OBSERVER POLE PLACEMENT VIA COMPANION FORM}
In $\RR / \CC$, observer design is dual to state feedback; over 
$\HH$ the same SISO workflow applies: transform to observable companion form, read the companion polynomial from the rightmost column, interpret poles as right-eigenvalue similarity classes, update coefficients in observer coordinates, and map the gain back. The construction is determinant-free, avoiding characteristic/minimal polynomials, Cayley–Hamilton, and adjugates.

\begin{The}[\emph{Pole placement via observable companion form}]
Assume a SISO observable pair \((A,C)\) with \(A\in\HH^{n\times n}\), \( C\in\HH^{1\times n}\). 
Let \((A_o,C_o)\) be its right observable companion form via an invertible \(T\in\HH^{n\times n}\), i.e.
\begin{IEEEeqnarray}{rCl}
A_o&=&T^{-1}AT, \quad C_o = CT = \bigl[0\ 0\ \cdots\ 1\bigr], \label{eq:obs-sim}
\end{IEEEeqnarray}
and write the rightmost column of \(A_o\) as \(\bigl[-a_0,\ldots,-a_{n-1}\bigr]^{\T}\).
Let
\begin{IEEEeqnarray*}{rCl}
a(\lambda)&=&a_0+a_1\lambda+\cdots+a_{n-1}\lambda^{n-1}+\lambda^n
\end{IEEEeqnarray*}
be the associated companion polynomial. Given any monic
\begin{IEEEeqnarray*}{rCl}
a_d(\lambda)&=&d_{0}+d_{1}\lambda+\cdots+d_{n-1}\lambda^{n-1}+\lambda^{n},
\end{IEEEeqnarray*}
there exists a unique observer gain \(L\in\HH^{n\times1}\) such that the
right spectrum of
\[
A_{\mathrm{obs}}:=A-LC
\]
consists exactly of the similarity classes induced by the right roots of \(a_d\),
counted with multiplicity.
Moreover, with \(L_o:=T^{-1}L\) we have
\begin{IEEEeqnarray}{rCl}
L &=& TL_o, \label{eq:L-backmap}\\
L_o &=& \bigl[d_0-a_0\ \ d_1-a_1\ \ \cdots\ \ d_{n-1}-a_{n-1}\bigr]^{\T}. \IEEEnonumber \label{eq:Lobserver-diff}
\end{IEEEeqnarray}
\end{The}

\begin{proof}
Work in observable coordinates. With \(L_o:=T^{-1}L\) and \eqref{eq:obs-sim},
\begin{IEEEeqnarray*}{rCl}
A_{\mathrm{obs},o}&:=&T^{-1}(A \!-\!LC)T  \!=\! A_o  \!-\! L_o C_o.
\end{IEEEeqnarray*}
Here $L_o C_o$ is the rank-one matrix whose last column equals $L_o$ and all other columns are zero. Writing \(L_o=\bigl[\ell_0,\ldots,\ell_{n-1}\bigr]^{\T}\), the rank–one update
subtracts \(L_o\) from the last (rightmost) column of \(A_o\). The rightmost column becomes $\bigl[-(a_0+\ell_0),\dots,-(a_{n-1}+\ell_{n-1})\bigr]^\T$
which yields the coefficient update.

A direct computation of the companion polynomial of \(A_{\mathrm{obs},o}\) shows that the closed–loop coefficients are updated as
\begin{IEEEeqnarray*}{rCl}
a_{\mathrm{obs}}(\lambda)&=&(a_0+\ell_0) + (a_1+\ell_1)\lambda 
 + \cdots \\ &&\negmedspace{}+ (a_{n-1}+\ell_{n-1})\lambda^{n-1} + \lambda^n.
\end{IEEEeqnarray*}
Choosing \(\ell_i=d_i-a_i\) makes \(a_{\mathrm{obs}}(\lambda)=a_d(\lambda)\).
By the observable–companion root/spectrum property, the right zeros of \(a_{\mathrm{obs}}\) are the right eigenvalues of \(A_{\mathrm{obs},o}\), hence of \(A_{\mathrm{obs}}\).
Uniqueness follows since the coefficient update \(\ell_i=d_i-a_i\) is bijective and \(L\leftrightarrow L_o\) via \(T\) is bijective. Finally \eqref{eq:L-backmap} holds by definition.
\end{proof}

\begin{Ex}[\emph{Example \ref{ExCompanion} continued}]\label{Ex0CompanionMethod} 
With the companion form \eqref{Eq0ExAoCo}, the similarity matrix \eqref{Eq0T}, and the companion polynomial \eqref{Eq0Excomppol} in hand,  we now consider several choices of desired polynomials:\\
\emph{(a) Real targets:}  $a_d(\lambda)=2+3\lambda+\lambda^2$ ($\lambda_1=-1,\lambda_2=-2$) gives
\begin{IEEEeqnarray*}{lCl}
    L_o&=&
    \begin{bmatrix}
    1.75+0.25\j \\
     2+0.5\j 
    \end{bmatrix}, \\
    L&=&
    \begin{bmatrix}
     1.25-2.25\j \\   
    -0.75\i+0.25\k
    \end{bmatrix},  \IEEEyesnumber \label{Eq0ExLreal}
\end{IEEEeqnarray*}
and the observer error dynamics matrix
\begin{IEEEeqnarray*}{rCl}
    A_{\rm obs} &=& 
    \begin{bmatrix}
    -2.75-\j    & 2.5\i-1.25\k \\
     0.5\i+0.75\k  &  -0.25-\j      
    \end{bmatrix}    \IEEEyesnumber \label{Eq0ExAobsreal}
\end{IEEEeqnarray*}
has right spectrum \(\sigma_\err(A_{\mathrm{obs}})=\{[-1],[-2]\}\), with both classes singleton.\\

\emph{(b) Complex pair:} 
 $a_d(\lambda)=2+2\lambda+\lambda^2$
 ($\lambda_{1,2}=-1\pm \i$) yields
\begin{IEEEeqnarray*}{lCl}
 L_o&=&
    \begin{bmatrix}
    1.75+0.25\j \\
     1+0.5\j   
    \end{bmatrix} ,    \\
 L&=&
    \begin{bmatrix}
     0.75-2.75\j   \\
    -0.25\i+1.75\k
    \end{bmatrix} , 
    \label{Eq0ExLcomplex} 
\end{IEEEeqnarray*}
and
\begin{IEEEeqnarray*}{rCl}\label{Eq0ExAobscomplex}
 A_{\rm obs} &=&
    \begin{bmatrix}
    -3.25-0.5\j  &  3\i-0.75\k  \\
     2\i+0.25\k  &    1.25-0.5\j
    \end{bmatrix} .   
\end{IEEEeqnarray*}
Here the assigned class is $\left[-1+\i\right]\in\sigma_\err(A_{\mathrm{obs}})$; equivalently, its standard representatives are $-1\pm\i$.
\\
\emph{(c) Quaternionic coefficients:} Choosing target similarity classes represented, for example, by $\lambda_1=-1+\j$ and $\lambda_2=-2+\k,$ sets (rounded)
\begin{IEEEeqnarray*}{rCl}
    a_d(\lambda)&=&
    2.7-\i-1.3\j+0.33\k \IEEEnonumber \\
    && + (3-0.67\i-0.33\j-0.33\k)\lambda + \lambda^2.
    \label{Eq0Exadlambdaquat}
  \end{IEEEeqnarray*}
Here \(a_d(\lambda)\in\mathbb{H}[\lambda]\) has nonreal (noncentral) coefficients, so a coordinate-free evaluation $a_d(A)$ is not similarity-equivariant, and we obtain (rounded)
\begin{IEEEeqnarray*}{rCl}
 L_o&=&
        \narrowbmatrix{ 
    2.42-\i-1.08\j+0.33\k \\
     2-0.67\i+0.17\j-0.33\k
     },
    \\
 L&=&  \narrowbmatrix{
 -1.25-1.167\i-3.75\j-1.83\k \\
    -1.5+1.42\i-1.17\j+1.75\k 
 },
\end{IEEEeqnarray*}
and (rounded)
\begin{IEEEeqnarray*}{l}
 A_{\rm obs} = \\
\quad \;
\tfrac12
    \narrowbmatrix{
         \!-\!4.2\!-\!1.8\i\!+\!1.5\j\!+\!1.2\k &
         \!-\!1.8\!+\!4\i\!-\!1.2\j\!+\!1.2\k  \\
    \!-\!1.2\!+\!2\i\!+\!1.5\j\!-\!1.4\k   &
    1.2\!+\!1.2\i\!+\!1.2\j\!+\!1.5\k
    }.
\end{IEEEeqnarray*}
Its right spectrum consists of the two similarity classes
\[
\sigma_\err(A_{\mathrm{obs}})=\{[-1+\i],[-2+\i]\}
=\{[-1+\j],[-2+\k]\}.
\]
Thus the observable-form coefficient update still assigns the desired similarity classes even when the coefficients are quaternionic.
\end{Ex}

\section{OBSERVER GAIN VIA ACKERMANN’S FORMULA}
Having established the quaternionic observer in its canonical (observable) form, we now turn to a compact algebraic method for directly computing the observer gain.
The resulting \emph{Observer Ackermann's formula} is the dual counterpart of the control-side expression. 
It provides a one–shot computation of \(L\) achieving any admissible placement of right–eigenvalue classes (observer poles).
Its structure mirrors the classical real/complex case. Yet, its justification over~$\HH$ requires special care: the argument avoids determinants and adjugates and instead exploits the observable companion relations and the right–polynomial calculus developed earlier.
This formulation establishes the quaternionic duality between controllability and observability in a determinant–free setting.
\begin{The}[\emph{Quaternionic Observer Ackermann’s formula}]
\label{ThAckermann-Obs}
Let \(A\in\HH^{n\times n}\), \(C\in\HH^{1\times n}\) be such that the SISO pair \((A,C)\) is observable and the observability matrix \(\mathcal{O}\) from \eqref{Eq0OBSmat} is invertible.
Let
\[ a_d(\lambda)=d_0+d_1\lambda+\cdots+d_{n-1}\lambda^{n-1}+\lambda^n\in\RR[\lambda]\]
be the desired monic error-dynamics polynomial with real coefficients, and let \(e_n=[0,0,\dots,1]^{\T}\in\RR^{n}\).
Then the observer gain column
\begin{IEEEeqnarray}{rCl}
L &=& a_d(A)\mathcal{O}^{-1} e_n,
\label{Eq0qacker0obs}
\end{IEEEeqnarray}
used in \eqref{eq:ct-fo} yields the error-dynamics matrix
\(A_{\mathrm{obs}}=A-LC\) whose right spectrum equals the similarity classes induced by the right roots
of \(a_d(\lambda)\):
\begin{IEEEeqnarray}{rCl}
\sigma_\err(A-LC) &=& \{[\lambda_i]: i=1,\ldots,n\},
\end{IEEEeqnarray}
where \(\{\lambda_i\}_{i=1}^n\) are right roots of \(a_d\) (listed with multiplicity).
\end{The}

\begin{Rem}
The restriction 
\[a_d\in\RR[\lambda]\] 
is essential: real coefficients lie in the center \(Z(\HH)=\RR\),
so \(a_d(A)\) is unambiguous and the required polynomial–similarity interchange is valid.
The assigned classes are real or complex-conjugate spherical:
for an \(n\)-th order SISO design, each real class consumes one degree and each nonreal class two,
so \(2s+r\le n\) when placing \(s\) nonreal and \(r\) real classes.
\end{Rem}

\begin{proof}
\emph{Observable companion coordinates.}
Choose \(T\) invertible with
\begin{IEEEeqnarray*}{rCl}
A_o &=& T^{-1}AT, \qquad C_o = CT = e_n^{\T}.
\end{IEEEeqnarray*}
Then \(A_o\) is in observable companion form with the rightmost column
\([ -a_0,\ldots,-a_{n-1}]^{\T}\) and companion polynomial
\(a(\lambda)=a_0+a_1\lambda+\cdots+a_{n-1}\lambda^{n-1}+\lambda^n\).
Let \(L_o:=T^{-1}L\). We have
\begin{IEEEeqnarray}{rCl}
T^{-1}(A-LC)T &=& A_o - L_o C_o = A_o - L_o e_n^{\T}.
\label{Eq:obs-closed-loop}
\end{IEEEeqnarray}

\emph{Coefficient matching in the right observable form.}
With \(u_k:=A_o^{k} e_1=e_{k+1}\) for \(k=0,\ldots,n-1\) and
\(A_o^{n} e_1 = -[\,a_0,\ldots,a_{n-1}\,]^{\T}\) (rightmost column read-off), we obtain
\begin{IEEEeqnarray}{rCl}
a_d(A_o)e_1
&=& \sum_{k=0}^{n-1} d_k A_o^{k} e_1 + A_o^{n} e_1 \IEEEnonumber\\
&=& \begin{bmatrix} d_0\\ \vdots\\ d_{n-1}\end{bmatrix}
  - \begin{bmatrix} a_0\\ \vdots\\ a_{n-1}\end{bmatrix}
= \begin{bmatrix} d_0-a_0\\ \vdots\\ d_{n-1}-a_{n-1}\end{bmatrix}.
\label{Eq:adAo-e1}
\end{IEEEeqnarray}
Setting \(L_o := a_d(A_o)e_1\) therefore adds \(L_o\) to the entries \(\{-a_k\}\) in the
\emph{rightmost column} of \(A_o\), producing the updated coefficients
\(\{- (a_k+\ell_k)\}\) and thus the closed-loop polynomial \(a_{\mathrm{obs}}(\lambda)=a_d(\lambda)\).
Since similarity preserves the right spectrum, the error matrix \(A-LC\) attains the desired
right-eigenvalue classes induced by \(a_d\).

\emph{Coordinate-free expression via \(\mathcal{O}\) and central polynomials.}
Let $\mathcal{O}_o$ be the observability matrix of $(A_o,C_o)$. Then
\begin{IEEEeqnarray*}{c}
\mathcal{O}
=
\begin{bmatrix}
C\\ CA\\ \vdots\\ CA^{n-1}
\end{bmatrix}
=
\begin{bmatrix}
C_o\\ C_o A_o\\ \vdots\\ C_o A_o^{n-1}
\end{bmatrix} T^{-1}
=
\mathcal{O}_o T^{-1},
\end{IEEEeqnarray*}
so $\mathcal{O}^{-1} = T\mathcal{O}_o^{-1}.$
Since $u_k:=A_o^k e_1$ satisfies $u_k=e_{k+1}$ for $k=0,\ldots,n-1$,
the first column of $\mathcal{O}_o$ is $e_n$, hence $\mathcal{O}_o e_1=e_n$ and
$\mathcal{O}_o^{-1} e_n = e_1$. Therefore
\[
\mathcal{O}^{-1} e_n = T e_1.
\]
Because $a_d\in\mathbb R[\lambda]$ and $A=T A_o T^{-1}$,
\[
a_d(A) = T\,a_d(A_o)\,T^{-1},
\]
and combining with $L=T L_o$ and $L_o=a_d(A_o)e_1$ gives
\[
L = T\,a_d(A_o)e_1 = a_d(A)\,T e_1 = a_d(A)\,\mathcal{O}^{-1} e_n,
\]
as claimed.

Similarity preserves the right spectrum, hence
\(\sigma_\err(A-LC)=\sigma_\err(A_o-L_oe_n^{\T})\),
and this spectrum consists exactly of the similarity classes
of the right zeros of \(a_d\).
\end{proof}

\begin{Ex}[\emph{Example \ref{Ex0CompanionMethod} continued}] \label{Ex0Acker}
Building on the previous example, recall $\mathcal{O}^{-1}$ from \eqref{Eq0OBSi}. We now consider \emph{real targets} only:  $a_d(\lambda)=2+3\lambda+\lambda^2$ ($\lambda_1=-1,\lambda_2=-2$) gives 
\begin{IEEEeqnarray*}{rCl}
    a_d(A)=
    \begin{bmatrix}
         0.625+0.5\j  &    0.5\i-0.125\k \\
     0.5\i+0.125\k &    0.625-0.5\j
    \end{bmatrix},
\end{IEEEeqnarray*} 
and the observer gain Ackermann formula \eqref{Eq0qacker0obs} yields
\begin{IEEEeqnarray*}{rCl}
    L&=&
    \begin{bmatrix}
     1.25-2.25\j   \\
    -0.75\i+0.25\k     
    \end{bmatrix}, \IEEEyesnumber \label{Eq0ExLreal2}
\end{IEEEeqnarray*}
which equals \eqref{Eq0ExLreal}. The resulting observer error dynamics matrix
$ A_{\rm obs} = A-LC$ is identical to \eqref{Eq0ExAobsreal} and 
its right spectrum is \(\sigma_\err(A_{\mathrm{obs}})=\{[-1],[-2]\}\), as desired, with both classes singleton.

For complex pair targets, the results coincide with Example \ref{Ex0CompanionMethod} (b).
By contrast, for quaternionic targets as in Example \ref{Ex0CompanionMethod} (c), the Ackermann formula returns an incorrect gain $L$ and misplaces the poles.
Coefficient update in companion coordinates can handle noncentral coefficients; the Ackermann shortcut cannot.
\end{Ex}

\section{SIMULATION EXAMPLE}
\begin{Ex}[\emph{based on previous examples}]
For \eqref{eq:ct-ss-siso} with \(A,C\) from \eqref{Eq0ACEx}, \(B=\begin{bmatrix}1&\!\j\end{bmatrix}^{\!\T}\), \(D=0\), input \(u(t)=1-\i+2\j-2\k\), and initial conditions \(x_0=\bigl[-1+\i-2\j+3\k,\;1+2\i-\j-2\k\bigr]^{\T}\), \(\hat x_0=\bigl[0,0\bigr]^{\T}\), the gain from \eqref{Eq0ExLreal}--\eqref{Eq0ExLreal2} yields convergence \(\hat x(t)\to x(t)\) in all quaternionic components; see Fig.~\ref{Fig01}.
\end{Ex}
\begin{figure}[!t]
    \centering    \includegraphics[width=\linewidth]{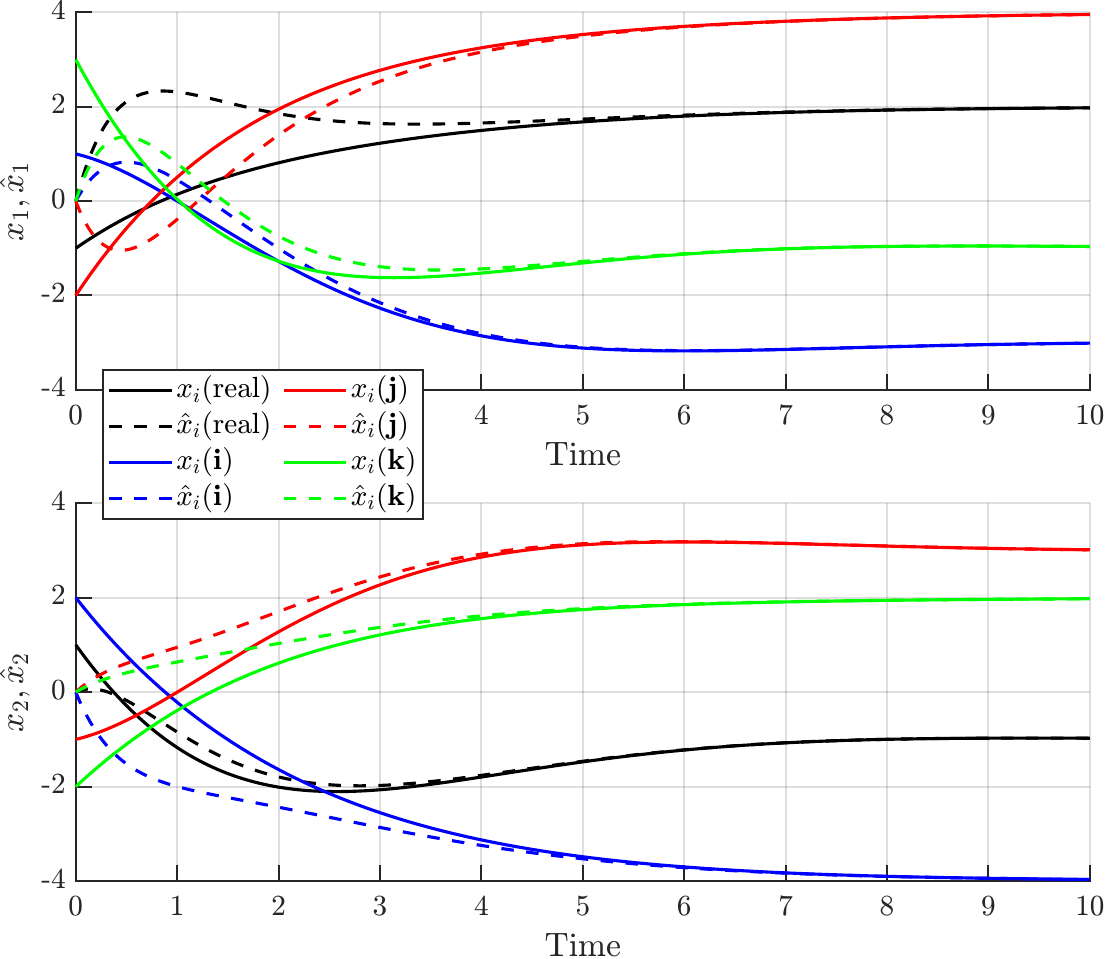}
    \caption{Observer states converge to system states in all quaternionic components.}
    \label{Fig01}
\end{figure}
%
\section{CONCLUSIONS}
We developed a characteristic-polynomial-free observer framework for quaternionic systems, centered on an observable companion realization and a quaternionic Ackermann formula for real-coefficient targets. The method assigns right-eigenvalue similarity classes directly in the error dynamics and makes the \(^*\)-based controller--observer duality explicit; companion-coordinate updates accommodate noncentral coefficients, whereas Ackermann does not. Application-oriented quaternionic observers for four-channel systems are discussed in~\cite{SebekEUROCAST2026}, while related controllability companion-form and state-feedback results are developed in~\cite{SebekTACsubmitted2025}.



\end{document}

%% file: macros.tex
\theoremstyle{definition}
\newtheorem{Def}{Definition}
\theoremstyle{definition}
\newtheorem{The}{Theorem}
\theoremstyle{definition}
\newtheorem{Lem}{Lemma}
\theoremstyle{definition}

\theoremstyle{definition}
\newtheorem{Pro}{Proposition}
\theoremstyle{definition}

\theoremstyle{remark}
\newtheorem{Rem}{Remark}
\theoremstyle{remark}
\newtheorem{Ex}{Example}
\theoremstyle{remark}

\theoremstyle{remark}
\newtheorem{Con}{Conjecture}
\theoremstyle{remark}

\theoremstyle{remark}

\usepackage{mdframed}

\DeclareFontEncoding{LS1}{}{}
\DeclareFontSubstitution{LS1}{stix2}{m}{n}
\DeclareMathAlphabet{\mathscr}{LS1}{stix2scr}{m}{n}
\newcommand{\err}{\ensuremath{ \mathscr{r} }}
\newcommand{\T}{\mathsf{T}}
\renewcommand{\i}{\mathbf{i}}
\renewcommand{\j}{\mathbf{j}}
\renewcommand{\k}{\mathbf{k}}

\renewcommand{\Re}{\operatorname{Re}}
\renewcommand{\Im}{\operatorname{Im}}

\newcommand{\NN}{\mathbb{N}}

\newcommand{\RR}{\mathbb{R}}
\newcommand{\CC}{\mathbb{C}}
\newcommand{\HH}{\mathbb{H}}

\newcommand{\HHnn}{\mathbb{H}^{n\times n}}

\newcommand{\HHmn}{\mathbb{H}^{m\times n}}
\newcommand{\DOI}[1]{\href{https://doi.org/#1}{DOI:\,\nolinkurl{#1}}}

%% file: arxiv_accepted_stamp.tex
\usepackage{ifthen}
\usepackage{eso-pic} 
\usepackage{xcolor}  

\newif\ifarxivstamp
\arxivstamptrue

\newcommand{\ArxivStampVenue}{for presentation at the 24th European Control Conference (ECC 2026), Reykjavik, Iceland.}
\newcommand{\ArxivStampDate}{March 23, 2026}
\newcommand{\ArxivStampFunding}{Co-funded by the European Union under the project ROBOPROX (reg. no. CZ.02.01.01/00/22\_008/0004590).}

\newcommand{\ArxivAcceptedStampOverlay}{%
  \ifarxivstamp
  \AddToShipoutPictureFG*{%
    \AtPageUpperLeft{%
      \raisebox{-1.6\baselineskip}[0pt][0pt]{%
        \begin{minipage}{\paperwidth}\centering\footnotesize
        \color{black}
        \hrule \vspace{2pt}
        \textbf{Accepted \ArxivStampVenue, \ArxivStampDate.}\\
        Preprint (accepted version). \ArxivStampFunding
        \vspace{2pt}\hrule
        \end{minipage}%
      }%
    }%
  }%
  \fi
}

% ...
